\documentclass{article}

\pdfoutput=1

\usepackage{style}

\usepackage{tikz}
\usetikzlibrary{positioning, calc}
\usetikzlibrary{calc}
\usetikzlibrary{arrows}
\usepackage{tikz-3dplot}
\usepackage{graphicx}
\usepackage{mdwlist}
\usepackage{color}
\usepackage{thmtools}
\usepackage{amssymb}
\usepackage{physics}
\usepackage{amsmath}
\usepackage{array}
\usetikzlibrary{fadings}
\usetikzlibrary{decorations.pathmorphing}
\usepackage{mathrsfs}
\usepackage{amsmath}

\usetikzlibrary{decorations.pathreplacing}
\tikzset{snake it/.style={decorate, decoration=snake}}

\usetikzlibrary{decorations.pathreplacing,decorations.markings}

\usepackage{slashed}

\usepackage{float}		% this is to place figures where requested!
\usepackage{graphicx}		% needed for the figures
\usepackage{subcaption}
\usepackage{bbold}

\newcommand{\CDS}{\textnormal{CDS}}
\newcommand{\CDQS}{\textnormal{CDQS}}
\newcommand{\pp}{\textnormal{pp}}
\newcommand{\pc}{\textnormal{pc}}
\newcommand{\cc}{\textnormal{cc}}
\newcommand{\QNP}{\textnormal{QNP}}
\newcommand{\coQNP}{\textnormal{coQNP}}
\newcommand{\NP}{\textnormal{NP}}
\newcommand{\coNP}{\textnormal{coNP}}
\newcommand{\FR}{\textnormal{FR}}
\newcommand{\FBB}{\textnormal{FBB}}

\usepackage{authblk}

\begin{document}

\title{Rank Lower Bounds on Non-local Quantum Computation} 

\author[1]{Vahid R. Asadi \thanks{ \texttt{vrasadi@uwaterloo.ca}}}

\author[1]{Eric Culf  \thanks{ \texttt{eculf@uwaterloo.ca}}}

\author[1,2]{Alex May \thanks{\texttt{amay@perimeterinstitute.ca}}}

\affil[1]{Institute for Quantum Computing, Waterloo, Ontario}
\affil[2]{Perimeter Institute for Theoretical Physics, Waterloo, Ontario}

\maketitle

\begin{abstract}
A non-local quantum computation (NLQC) replaces an interaction between two quantum systems with a single simultaneous round of communication and shared entanglement. 
We study two classes of NLQC, $f$-routing and $f$-BB84, which are of relevance to classical information theoretic cryptography and quantum position-verification. 
We give the first non-trivial lower bounds on entanglement in both settings, but are restricted to lower bounding protocols with perfect correctness.
Within this setting, we give a lower bound on the Schmidt rank of any entangled state that completes these tasks for a given function $f(x,y)$ in terms of the rank of a matrix $g(x,y)$ whose entries are zero when $f(x,y)=0$, and strictly positive otherwise.
This also leads to a lower bound on the Schmidt rank in terms of the non-deterministic quantum communication complexity of $f(x,y)$. 
Because of a relationship between $f$-routing and the conditional disclosure of secrets (CDS) primitive studied in information theoretic cryptography, we obtain a new technique for lower bounding the randomness complexity of CDS. 
\end{abstract}

\tableofcontents

%%%%%%%%%%%%%%%%%%%%%%%%%%%%%%%%%%%%%%%%%%%%%%%%%%%%%%%%
\section{Introduction}
%%%%%%%%%%%%%%%%%%%%%%%%%%%%%%%%%%%%%%%%%%%%%%%%%%%%%%%%

A non-local quantum computation (NLQC) replaces a local interaction with entanglement and a single simultaneous round of communication. 
See \cref{fig:non-localandlocal}.
NLQC has applications in a number of areas, including quantum position-verification \cite{kent2006tagging,kent2011quantum,malaney2016quantum}, Hamiltonian simulation \cite{apel2024security}, the AdS/CFT correspondence \cite{may2019quantum, may2020holographic,may2021holographic, may2022complexity}, and information theoretic cryptography \cite{allerstorfer2023relating}. 
In all of these settings, we are interested in understanding the resources necessary to implement a given computation in the non-local form. 
Lower bounds on the entanglement necessary to complete a non-local computation are so far poorly understood but are central to these applications. 

Two important classes of non-local computations are $f$-routing and $f$-BB84, both initially studied in the context of quantum position-verification \cite{kent2011quantum,bluhm2021position}. 
These computations involve classical inputs of size $n$ and $O(1)$ size quantum inputs. 
The local implementation of these computations involves computing a classical function $f$ from the classical inputs, then doing $O(1)$ quantum operations conditioned on the outcome. 
A longstanding goal has been to prove that the non-local implementation of the same computation requires growing quantum resources, and in particular growing entanglement. 
Recently, a bound on the quantum system size was proven in \cite{bluhm2021position}, but this assumes a purified view that absorbs classical processing into the quantum part of the computation.\footnote{See \cite{asadi2024linear} for more discussion of this point.} 
Another approach appears in \cite{asadi2024linear}, which bounds the number of quantum gates required in the non-local implementation. 
Here, we achieve the goal of lower bounding entanglement, at least as quantified by the Schmidt rank of the shared system, but at the cost of only considering protocols that complete $f$-routing or $f$-BB84 with only one-sided error.

\begin{figure*}
    \centering
    \begin{subfigure}{0.45\textwidth}
    \centering
    \begin{tikzpicture}[scale=0.6]
    
    %interaction unitary unitary
    \draw[thick] (-1,-1) -- (-1,1) -- (1,1) -- (1,-1) -- (-1,-1);
    
    %wire into interaction unitary
    \draw[thick] (-3.5,-3) to [out=90,in=-90] (-0.5,-1);
    \node[below] at (-3.5,-3) {$A$};
    \draw[thick] (3.5,-3) to [out=90,in=-90] (0.5,-1);
    \node[below] at (3.5,-3) {$B$};
    
    \draw[thick] (0.5,1) to [out=90,in=-90] (3.5,3);
    \node[above] at (3.5,3) {$B$};
    \draw[thick] (-0.5,1) to [out=90,in=-90] (-3.5,3);
    \node[above] at (-3.5,3) {$A$};
    
    \node at (0,0) {$\mathcal{T}$};

    \node at (0,-5) {$ $};
    
    \end{tikzpicture}
    \caption{}
    \label{fig:local}
    \end{subfigure}
    \hfill
    \begin{subfigure}{0.45\textwidth}
    \centering
    \begin{tikzpicture}[scale=0.4]
    
    %lower left box
    \draw[thick] (-5,-5) -- (-5,-3) -- (-3,-3) -- (-3,-5) -- (-5,-5);
    \node at (-4,-4) {$\mathcal{V}^L$};
    
    %lower right box
    \draw[thick] (5,-5) -- (5,-3) -- (3,-3) -- (3,-5) -- (5,-5);
    \node at (4,-4) {$\mathcal{V}^R$};
    
    %top right box
    \draw[thick] (5,5) -- (5,3) -- (3,3) -- (3,5) -- (5,5);
    \node at (4,4) {$\mathcal{W}^R$};
    
    %top left box
    \draw[thick] (-5,5) -- (-5,3) -- (-3,3) -- (-3,5) -- (-5,5);
    \node at (-4,4) {$\mathcal{W}^L$};
    
    %left vertical wire
    \draw[thick] (-4.5,-3) -- (-4.5,3);
    
    %right vertical wire
    \draw[thick] (4.5,-3) -- (4.5,3);
    
    %left to right wire
    \draw[thick] (-3.5,-3) to [out=90,in=-90] (3.5,3);
    
    %right to left wire
    \draw[thick] (3.5,-3) to [out=90,in=-90] (-3.5,3);
    
    %entanglement
    \draw[thick] (-3.5,-5) to [out=-90,in=-90] (3.5,-5);
    \draw[black] plot [mark=*, mark size=3] coordinates{(0,-7.05)};
    
    %input wires
    \draw[thick] (-4.5,-6) -- (-4.5,-5);
    \node[below] at (-4.5,-6) {$A$};
    \draw[thick] (4.5,-6) -- (4.5,-5);
    \node[below] at (4.5,-6) {$B$};
    
    %output wires
    \draw[thick] (4.5,5) -- (4.5,6);
    \node[above] at (-4.5,6) {$A$};
    \draw[thick] (-4.5,5) -- (-4.5,6);
    \node[above] at (4.5,6) {$B$};
    
    \end{tikzpicture}
    \caption{}
    \label{fig:non-localcomputation}
    \end{subfigure}
    \caption{Local and non-local computations. a) A channel $\mathcal{T}_{AB\rightarrow AB}$ is implemented by directly interacting the input systems. b) A non-local quantum computation. The goal is for the action of this circuit on the $AB$ systems to approximate the channel $\mathcal{T}_{AB\rightarrow AB}$.}
    \label{fig:non-localandlocal}
\end{figure*}

It is a striking feature of non-local quantum computation that these mostly classical local computations require large quantum entanglement to implement non-locally. 
Beyond this, establishing this quantum-classical separation is relevant for quantum position-verification and for information theoretic cryptography. 
In a QPV scheme, we want to find operations that are as easy as possible for the honest player (who acts locally) to implement, but as hard as possible for the dishonest player (who implements an NLQC). 
Proving entanglement lower bounds for $f$-routing or $f$-BB84 gives a verification scheme where the honest player implements a nearly classical computation, while the dishonest player needs large quantum resources. 
This establishes the security of such QPV schemes under a bounded entanglement assumption. 
Our bound is only in the setting of one-sided error and hence not immediately relevant to practical schemes, but represents a step towards proving entanglement lower bounds in the robust setting.\footnote{Many entanglement lower bounds apply in the robust setting, see e.g. \cite{tomamichel2013monogamy, escola2024quantum, may2022complexity, beigi2011simplified}, but none grow with the classical input size. We achieve growth with classical input size at the expense of robustness.}

Also relevant to our work is the connection between NLQC and information theoretic cryptography. 
In information theoretic cryptography, we wish to understand the communication and randomness cost to establish information theoretic privacy in various cryptographic settings. 
A primitive known as conditional disclosure of secrets (CDS) has been identified as one of the simplest settings where one can hope to understand the cost of privacy. 
CDS was initially studied in the context of private information retrieval \cite{GERTNER2000592}, has applications in secret sharing \cite{applebaum2020power}, and shares a number of connections to communication complexity and interactive proofs \cite{applebaum2021placing}. In \cite{allerstorfer2023relating}, $f$-routing was shown to be closely related to CDS: lower bounds on entanglement cost in $f$-routing give lower bounds on randomness in CDS. 
This opens a new avenue for quantum information techniques to make contributions to classical information theoretic cryptography. 

\subsection{Summary of our results}
An $f$-routing task is defined by a choice of Boolean function $f:\{0,1\}^n\times \{0,1\}^n\rightarrow \{0,1\}$. 
In the task, Alice, on the left, receives a quantum system $Q$ of $O(1)$ size and a classical string $x\in \{0,1\}^n$. 
Bob, on the right, receives $y\in \{0,1\}^n$. 
Their goal is to bring $Q$ to Alice if $f(x,y)=0$, and to Bob if $f(x,y)=1$. 
In the setting of non-local computation, Alice and Bob may only use shared entanglement and a single round of communication to accomplish this. 
A difficulty in accomplishing this task arises since quantum mechanics prevents the system in register~$Q$ from being copied, and hence apparently Alice is forced to decide alone whether to keep or send~$Q$. 
Furthermore, she doesn't know where~$Q$ is needed until learning $y$, at which point the single round of communication has already passed. 

Despite the difficulty, the $f$-routing task can be completed for arbitrary functions $f$, with the most efficient known protocols using $2^{\Theta(\sqrt{n\log n})}$ entanglement and communication for arbitrary functions \cite{allerstorfer2023relating}, and efficient protocols known for some low complexity classes of functions \cite{cree2023code,buhrman2013garden}. 
So far, only an $O(1)$ lower bound is known on entanglement in $f$-routing, as we review in \cref{appendix:trivialbounds}.
Here, we prove a new lower bound on $f$-routing with one-sided perfection. 
We define a $(\epsilon_0,\epsilon_1)$-correct $f$-routing scheme as one which brings $Q$ to the left with fidelity at least $1-\epsilon_0$ when $(x,y)\in f^{-1}(0)$ and brings $Q$ to the right with fidelity at least $1-\epsilon_1$ when $(x,y)\in f^{-1}(1)$. 
Define $FR_0(f)$ as the log of the minimal Schmidt rank of any entangled state that suffices to complete $f$-routing with $\epsilon_0=0$, $\epsilon_1>0$ for the function $f$, and $FR_1(f)$ as the analogous quantity with $\epsilon_0>0, \epsilon_1=0$. 
We prove the bounds
\begin{align}\label{eq:introfRbounds}
    \FR_0(f) &\geq \frac{1}{4}\log \rank(g|_f) \enspace, \nonumber \\
    \FR_1(f) &\geq \frac{1}{4}\log \rank(g|_{\neg f})\enspace,
\end{align}
where $g|_f(x,y)$ is a real-valued function that satisfies $f(x,y)=0 \leftrightarrow g(x,y)=0$, $f(x,y)=1\leftrightarrow g(x,y)>0$. 
For some natural functions, this provides explicit strong lower bounds. 
For instance, taking $f=EQ$ to be the equality function, we are assured $g(x,y)$ is diagonal with non-zero entries on the diagonal and hence is of full rank, therefore $FR_0(EQ)\geq n/4$. 
Similarly, we obtain linear lower bounds on $FR_1(NEQ)$ where $NEQ$ is the non-equality function. 

As mentioned above, $f$-routing is related closely to the conditional disclosure of secrets (CDS) primitive. 
A CDS instance is defined by a choice of Boolean function $f:\{0,1\}^n\times \{0,1\}^n\rightarrow \{0,1\}$. 
In CDS, Alice and Bob receive inputs $x\in \{0,1\}^n$ and $y\in \{0,1\}^n$ respectively, while Alice additionally holds a secret string $s\in\{0,1\}^k$.
Alice and Bob share randomness, and send messages to a referee. 
Their goal is for the referee to be able to determine $s$ if and only if $f(x,y)=1$.  

In \cite{allerstorfer2023relating}, it was shown that the randomness cost of CDS for a function $f$ upper bounds the entanglement cost of an $f$-routing scheme for the same function. 
Because of this, our lower bounds on entanglement in $f$-routing imply the same lower bounds on randomness in CDS. 
Because we only need one sided perfection for our bounds to apply, it is possible to bound CDS with imperfect correctness and perfect privacy (ppCDS), or perfect correctness but imperfect privacy (pcCDS). We obtain:
\begin{align*}
    \pp{\CDS}(f) &\geq \Omega(\log \rank(g|_f))\enspace, \nonumber \\
    \pc{\CDS}(f) &\geq \Omega(\log \rank(g|_{\neg f}))\enspace.
\end{align*}
The expressions on the left hand side denote randomness costs. 
This provides a new technique for bounding randomness in CDS, which is proven using the quantum information perspective but can be phrased purely in terms of classical quantities.

We also prove lower bounds against the $f$-BB84 class of non-local quantum computations. 
An $f$-BB84 task is defined by a choice of Boolean function $f(x,y)$ with $x\in \{0,1\}^n$ given to Alice and $y\in \{0,1\}^n$ given to Bob. 
The system $Q$ is in one of the BB84 states $\{H^{f(x,y)}\ket{b}\}$, with $b\in \{0,1\}$ and $f(x,y)\in\{0,1\}$. 
Alice and Bob should both output $b$. 

The $f$-BB84 task can be completed for arbitrary functions $f$, with the most efficient known protocols using $2^{\Theta(n)}$ entanglement and communication for arbitrary functions \cite{buhrman2013garden}.\footnote{Note that \cite{buhrman2013garden} deals with the $f$-routing task, but their protocol is easily adapted to $f$-BB84.} 
Again, we consider separate error rates in $0$ and $1$ instances; a $(\epsilon_0, \epsilon_1)$-correct $f$-BB84 protocol produces $b$ on both sides with probability $1-\epsilon_0$ when $f(x,y)=0$, and produces $b$ on both sides with probability $1-\epsilon_1$ when $f(x,y)=1$. 
We prove lower bounds on the log of the Schmidt rank of any resource state that suffices to complete $f$-BB84 with one-sided error, so that one of $\epsilon_0$ or $\epsilon_1$ is zero. 
In particular, (nearly) the same bound as in the $f$-routing case applies, 
\begin{align*}
    \FBB84_{0}(f) \geq \frac{1}{4}(\log \rank(g|_f)-1)\enspace,\nonumber \\
    \FBB84_{1}(f) \geq \frac{1}{4}(\log \rank(g|_{\neg f})-1)\enspace,
\end{align*}
where again $g$ is a function with zeros exactly at the zeros of $f$ and is positive otherwise, and $\FBB84_{i}(f)$ denotes the Schmidt rank of a resource system needed to complete $f$-BB84 with $\epsilon_i=0$, $i\in\{0,1\}$.
In a practical setting, $f$-BB84 is of particular interest because it can be made loss-tolerant~\cite{allerstorfer2023making}, and doesn't require quantum communication from the honest prover to the verifier. 
Unfortunately, our bound does not carry over to the noisy setting, so is not of immediate practical value.
Nonetheless, our bound is the first non-trivial entanglement bound on $f$-BB84.
For $f$-BB84, the bound above leads to linear lower bounds for the equality, non-equality, set-disjointness, and greater-than functions. 

We can better understand our lower bounds by relating them to communication complexity. 
In a quantum communication complexity protocol, Alice holds $x\in \{0,1\}^n$, Bob holds $y\in \{0,1\}^n$, and Alice and Bob communicate qubits to determine $f(x,y)$. 
The number of qubits of communication needed for Alice and Bob to output $1$ with non-zero probability when $f(x,y)=1$, and never output $1$ when $f(x,y)=0$ is known as the quantum non-deterministic communication complexity, $\QNP^{cc}(f)$. 
Using the characterization of $\QNP^{cc}$ in terms of the non-deterministic rank \cite{de2000characterization,de2003nondeterministic}, we obtain 
\begin{align*}
    \log \rank g|_{f} &\geq (\QNP^{\cc}(f) -1)\enspace, \nonumber \\
    \log \rank g|_{\neg f} &\geq (\coQNP^{\cc}(f) - 1)\enspace.
\end{align*}
This relates our lower bounds on CDS, $f$-routing, and $f$-BB84 to quantum communication complexity.

%%%%%%%%%%%%%%%%%%%%%%%%%%%%%%%%%%%%%%%%%%%%%%%%%%%%%%%%
\section{Preliminaries}
%%%%%%%%%%%%%%%%%%%%%%%%%%%%%%%%%%%%%%%%%%%%%%%%%%%%%%%%

%%%%%%%%%%%%%%%%%%%%%%%%%%%%%%%%%%%%%%%%%%%%%%%%%%%%%%%%
\subsection{Notation and definitions}
%%%%%%%%%%%%%%%%%%%%%%%%%%%%%%%%%%%%%%%%%%%%%%%%%%%%%%%%

We let $\Psi^+_{AB}$ denote the maximally entangled state on $AB$, so that
\begin{align*}
    \ket{\Psi^+}_{AB}=\frac{1}{\sqrt{d}}\sum_{i=1}^{d} \ket{i}_A\ket{i}_B \enspace,
\end{align*}
and $\Psi^+_{AB} = \ketbra{\Psi^+}{\Psi^+}_{AB}$. 

Given two density matrices $\rho,\sigma$, we define the fidelity by
\begin{align*}
    F(\rho,\sigma) = \tr \sqrt{\sqrt{\sigma}\rho\sqrt{\sigma}}\enspace.
\end{align*} 
The fidelity is related to the one-norm distance by the Fuchs--van de Graaf inequalities, 
\begin{align*}
    1-F(\rho,\sigma) \leq \frac{1}{2}||\rho-\sigma||_1 \leq \sqrt{1- (F(\rho,\sigma))^2}\enspace.
\end{align*}

The von Neumann entropy is defined by
\begin{align*}
    S(A) = -\tr (\rho\log \rho)\enspace.
\end{align*}
The quantum mutual information is defined by
\begin{align*}
    I(A:B)_{\rho} = S(A)_\rho + S(B)_\rho - S(AB)_\rho\enspace,
\end{align*}
and the conditional entropy is defined by
\begin{align*}
    S(A|B)_\rho = S(AB)_\rho - S(B)_\rho\enspace.
\end{align*}
We will also make use of the following statement of continuity of the conditional entropy \cite{winter2016tight}, 
\begin{theorem}[\cite{winter2016tight}]\label{eq:condentropycont}
Suppose that
\begin{align*}
    \frac{1}{2}||\rho_{AB}-\sigma_{AB}||_1 \leq \epsilon\enspace,
\end{align*}
and let $h(x)=-x\log x - (1-x)\log (1-x)$ be the binary entropy function. Then
\begin{align*}
    |S(A|B)_\rho - S(A|B)_\sigma| \leq 2\epsilon \log d_A + (1+\epsilon) h\left(\frac{\epsilon}{1+\epsilon}\right)\enspace.
\end{align*}
\end{theorem}

We will also make use of the following observation.

\begin{observation}\label{clm:trace-sq}
Let $\rho_{AB}$ be a density matrix on the $AB$ system, $d$ be the dimension of subsystem $A$, and $\rho_B$ be the marginal of $\rho_{AB}$ on the subsystem $B$. Then we have
    \begin{align*}
        \tr(\rho_{AB}-\frac{\mathcal{I}_A}{d}\otimes\rho_B)^2
        &= \tr(\rho^2_{AB} +\frac{\mathcal{I}_A}{d^2}\otimes\rho^2_B - 2\rho_{AB}(\frac{\mathcal{I}_A} {d}\otimes \rho_B)) \\
        &= \tr(\rho^2_{AB})+\frac{1}{d^2}\tr(\mathcal{I}_A)\tr(\rho^2_B)-\frac{2}{d}\tr_B(\rho^2_B) \\
        &= \tr(\rho^2_{AB}) - \frac{1}{d}\tr(\rho^2_B) \\
        &= \tr(\rho^2_{AB}) - \tr(\frac{\mathcal{I}_A}{d^2}\otimes \rho^2_B) \enspace.
    \end{align*}
\end{observation}

\subsection{Communication complexity}
We need some definitions and results from communication complexity. 
Consider two parties, Alice and Bob, with Alice given input $x\in \{0,1\}^n$ and Bob given input $y\in \{0,1\}^n$. 
Alice and Bob exchange messages and Alice outputs a bit $z$. 
The \emph{non-deterministic quantum communication complexity} $\QNP^{cc}(f)$, introduced in \cite{de2000characterization}, is the minimal number of qubits Alice and Bob need to exchange to both output $1$ with probability larger than zero when $f(x,y)=1$ and both output $0$ with probability $1$ when $f(x,y)=0$. 

A matrix $M$ is called a non-deterministic communication complexity matrix for $f$ exactly when $M(x,y)\neq 0$ when $f(x,y)=1$. 
The non-zero entries are allowed to be any complex numbers. 
The minimal rank of any non-deterministic communication complexity matrix for $f$ is denoted by $nrank(f)$ and called the non-deterministic rank.
We have the following result from \cite{de2003nondeterministic}. 

\begin{theorem}[\cite{de2003nondeterministic}]\label{thm:QNPandnrank}
The non-deterministic rank and the non-deterministic quantum communication complexity are related by
    \begin{align*}
        \QNP^{\cc}(f) = \lceil{ \log (\textnormal{nrank}(f))\rceil} + 1\enspace.
    \end{align*}
\end{theorem}

%%%%%%%%%%%%%%%%%%%%%%%%%%%%%%%%%%%%%%%%%%%%%%%%%%%%%%%%
\section{Lower bounds from function rank}
%%%%%%%%%%%%%%%%%%%%%%%%%%%%%%%%%%%%%%%%%%%%%%%%%%%%%%%%

Below, we first define and analyzing the $f$-routing task, and then prove our result in this setting.
In the subsequent section we show the same lower bounds for $f$-BB84.

%%%%%%%%%%%%%%%%%%%%%%%%%%%%%%%%%%%%%%%%%%%%%%%%%%%%%%%%
\subsection{\texorpdfstring{$f$}{TEXT}-routing bounds from function rank}
%%%%%%%%%%%%%%%%%%%%%%%%%%%%%%%%%%%%%%%%%%%%%%%%%%%%%%%%

\begin{figure*}
    \centering
    \begin{subfigure}{0.45\textwidth}
    \centering
    \begin{tikzpicture}[scale=0.44]
    
    %lower left box
    \draw[thick] (-5,-5) -- (-5,-3) -- (-3,-3) -- (-3,-5) -- (-5,-5);
    \node at (-4,-4) {$\mathcal{N}^x$};
    
    %lower right box
    \draw[thick] (5,-5) -- (5,-3) -- (3,-3) -- (3,-5) -- (5,-5);
    \node at (4,-4) {$\mathcal{M}^y$};
    
    %top right box
    \draw[thick] (5,5) -- (5,3) -- (3,3) -- (3,5) -- (5,5);
    \node at (4,4) {$\mathcal{W}^R$};
    
    %top left box
    \draw[thick] (-5,5) -- (-5,3) -- (-3,3) -- (-3,5) -- (-5,5);
    \node at (-4,4) {$\mathcal{W}^L$};
    
    %left vertical wire
    \draw[thick] (-4.5,-3) -- (-4.5,3);
    \node[left] at (-4.5,-2) {$M_0$};
    
    %right vertical wire
    \draw[thick] (4.5,-3) -- (4.5,3);
    \node[right] at (4.5,-2) {$M_1'$};
    
    %left to right wire
    \draw[thick] (-3.5,-3) to [out=90,in=-90] (3.5,3);
    \node[right] at (-3.25,-2) {$M_0'$};
    
    %right to left wire
    \draw[thick] (3.5,-3) to [out=90,in=-90] (-3.5,3);
    \node[left] at (3.25,-2) {$M_1$};
    
    %entanglement
    \draw[thick] (-3.5,-5) to [out=-90,in=-90] (3.5,-5);
    \draw[black] plot [mark=*, mark size=3] coordinates{(0,-7.05)};
    
    %input wires
    \draw[thick] (-4.5,-6) -- (-4.5,-5);
    \node[below] at (-4.5,-6) {$Q,x$};
    \draw[thick] (4.5,-6) -- (4.5,-5);
    \node[below] at (4.5,-6) {$y$};
    
    %output wires
    \draw[thick] (4.5,5) -- (4.5,6);
    \node[above] at (-4.5,6) {$A$};
    \draw[thick] (-4.5,5) -- (-4.5,6);
    \node[above] at (4.5,6) {$B$};
    
    \end{tikzpicture}
    \caption{}
    \label{fig:fR}
    \end{subfigure}
    \hfill
    \begin{subfigure}{0.45\textwidth}
    \centering
    \begin{tikzpicture}[scale=0.44]
    
    %lower left box
    \draw[thick] (-5,-5) -- (-5,-3) -- (-3,-3) -- (-3,-5) -- (-5,-5);
    \node at (-4,-4) {$\mathcal{N}^x$};
    
    %lower right box
    \draw[thick] (5,-5) -- (5,-3) -- (3,-3) -- (3,-5) -- (5,-5);
    \node at (4,-4) {$\mathcal{M}^y$};
    
    %top right box
    \draw[thick] (5.5,5) -- (5.5,3) -- (2.5,3) -- (2.5,5) -- (5.5,5);
    \node at (4,4) {\small{$\{\Lambda^{x,y,i}_{M'}\}_i$}};
    
    %top left box
    \draw[thick] (-5.5,5) -- (-5.5,3) -- (-2.5,3) -- (-2.5,5) -- (-5.5,5);
    \node at (-4,4) {\small{$\{\Lambda^{x,y,i}_M\}_i$}};
    
    %left vertical wire
    \draw[thick] (-4.5,-3) -- (-4.5,3);
    \node[left] at (-4.5,-2) {$M_0$};
    
    %right vertical wire
    \draw[thick] (4.5,-3) -- (4.5,3);
    \node[right] at (4.5,-2) {$M_1'$};
    
    %left to right wire
    \draw[thick] (-3.5,-3) to [out=90,in=-90] (3.5,3);
    \node[right] at (-3.25,-2) {$M_0'$};
    
    %right to left wire
    \draw[thick] (3.5,-3) to [out=90,in=-90] (-3.5,3);
    \node[left] at (3.25,-2) {$M_1$};
    
    %entanglement
    \draw[thick] (-3.5,-5) to [out=-90,in=-90] (3.5,-5);
    \draw[black] plot [mark=*, mark size=3] coordinates{(0,-7.05)};
    
    %input wires
    \draw[thick] (-4.5,-6) -- (-4.5,-5);
    \node[below] at (-4.5,-6) {$Q,x$};
    \draw[thick] (4.5,-6) -- (4.5,-5);
    \node[below] at (4.5,-6) {$y$};
    
    %output wires
    \draw[thick] (4.5,5) -- (4.5,6);
    \node[above] at (-4.5,6) {$b'$};
    \draw[thick] (-4.5,5) -- (-4.5,6);
    \node[above] at (4.5,6) {$b''$};
    
    \end{tikzpicture}
    \caption{}
    \label{fig:fBB84}
    \end{subfigure}
    \caption{a) General NLQC implementing a $f$-routing position-verification scheme. The first round operations are quantum channels that depend on the inputs $x\in \{0,1\}^n$ and $y\in\{0,1\}^n$. In the second round, Alice and Bob apply channels mapping to qubit systems $A$, $B$. If $f(x,y)=0$ $A$ should be maximally entangled with the reference $\bar{Q}$. If $f(x,y)=1$ then $B$ should be maximally entangled with $R$. b) Non-local quantum computation implementing a $f$-BB84 position-verification scheme. The first round operations are quantum channels that depend on the inputs $x\in \{0,1\}^n$ and $y\in\{0,1\}^n$. They communicate in one simultaneous exchange, and then apply measurements to their local systems. They succeed if $b=b'=b''$, with $b$ determined by measuring a reference maximally entangled with $Q$.}
    \label{fig:fBB84andfR}
\end{figure*}

We begin with a definition of an $f$-routing task. 

\begin{definition}\label{def:qubitfrouting}
    A \textbf{qubit $f$-routing} task is defined by a choice of Boolean function $f:\{ 0,1\}^{2n}\rightarrow \{0,1\}$, and a $2$ dimensional Hilbert space $\mathcal{H}_Q$.
    Inputs $x\in \{0,1\}^{n}$ and system $Q$ are given to Alice, and input $y\in \{0,1\}^{n}$ is given to Bob.
    Alice and Bob exchange one round of communication, with the combined systems received or kept by Bob labelled $M'$ and the systems received or kept by Alice labelled $M$.
    Label the combined actions of Alice and Bob in the first round as $\mathcal{N}^{x,y}_{Q\rightarrow MM'}$. 
    The qubit $f$-routing task is completed $(\epsilon_0,\epsilon_1)$-correctly on an input $(x,y)$ if there exists a channel $\mathcal{D}^{x,y}_{M\rightarrow Q}$ such that,
    \begin{align*}
         \text{when} \,\,\, f(x,y)=0,\,\,\, F(\mathcal{D}^{x,y}_{M\rightarrow Q} \circ\tr_{M'} \circ\mathcal{N}^{x,y}_{Q\rightarrow MM'}(\Psi^+_{\bar{Q}Q}),\Psi^+_{\bar{Q}Q}) \geq 1-\epsilon_0\enspace,
    \end{align*}
    and there exists a channel $\mathcal{D}^{x,y}_{M'\rightarrow Q}$ such that
    \begin{align*}
        \text{when}\,\,\, f(x,y)=1,\,\,\,F(\mathcal{D}^{x,y}_{M'\rightarrow Q} \circ\tr_{M} \circ\mathcal{N}^{x,y}_{Q\rightarrow MM'}(\Psi^+_{\bar{Q}Q}),\Psi^+_{\bar{Q}Q}) \geq 1-\epsilon_1\enspace.
    \end{align*}
    In words, Alice can recover $Q$ if $f(x,y)=0$ and Bob can recover $Q$ if ${f(x,y)=1}$. 
\end{definition}
See the system labels shown in \cref{fig:fR}. 
Note that we label $M_0M_1=M$, $M_0'M_1'=M'$. 

We will focus on $f$-routing with $\epsilon_0=0$, $\epsilon_1\geq 0$ so that the protocol is perfectly correct on zero instances of $f(x,y)$, or with $\epsilon_0>0$ and $\epsilon_1=0$ so that the protocol is perfectly correct on $1$ instances. 
We define the entanglement cost of an $f$-routing protocol to be the log of the minimal Schmidt rank of any resource system that can be used to perform the $f$-routing task. 
In notation, we define $\FR_0(f)$ to be entanglement cost for $f$-routing with $\epsilon_0=0$, $\epsilon_1=0.05$, $\FR_1(f)$ to be the entanglement cost when $\epsilon_0=0.05$, $\epsilon_1=0$.

Towards a lower bound in this setting, we first observe that $\rho_{\bar{Q}M'}=\rho_{\bar{Q}}\otimes \rho_{M'}$ iff $f(x,y)=0$. We show this in the following lemma.
\begin{lemma}\label{lemma:f=0tensorproduct}
    Suppose an $f$-routing protocol is perfectly correct on zero instances, and $\epsilon_1<\epsilon_1^*$ correct on $1$ instances, where $\epsilon_1^*$ depends on $d_Q$ but is lower bounded by a constant. Then the mid-protocol state it produces, $\rho_{\bar{Q}M'}$, satisfies $\rho_{\bar{Q}M'}=\rho_{\bar{Q}}\otimes \rho_{M'}$ if and only if $f(x,y)=0$.
\end{lemma}
\begin{proof}\,
    By correctness on $0$ instances, we have that 
    \begin{align*}
        \text{when}\,\,\, f(x,y)=0,\,\,\,F(\mathcal{D}^{x,y}_{M\rightarrow Q} \circ\tr_{M'} \circ\mathcal{N}^{x,y}_{Q\rightarrow MM'}(\Psi^+_{\bar{Q}Q}),\Psi^+_{\bar{Q}Q}) =1\enspace,
    \end{align*}
    hence $f(x,y)=0$ implies Alice can produce a Bell state $\ket{\Psi^+}_{\bar{Q}Q}$. 
    By data processing we have
    \begin{align*}
        2n=I(\bar{Q}:Q)_{\Psi^+_{\bar{Q}Q}} \leq I(\bar{Q}:M)_{\rho_{\bar{Q}M}}\enspace, 
    \end{align*}
    but we also have the inequality
    \begin{align*}
        I(\bar{Q}:M)_\rho + I(\bar{Q}:M')_\rho \leq 2S(\bar{Q})\enspace,
    \end{align*}
    so that $I(\bar{Q}:M')=0$, which implies
    \begin{align*}
        \rho_{\bar{Q}M'}=\rho_{\bar{Q}}\otimes \rho_{M'}\enspace,
    \end{align*}
    as needed. 
    Conversely, if $f(x,y)=1$ we cannot have $\rho_{\bar{Q}M'}=\rho_{\bar{Q}}\otimes \rho_{M'}$. 
    To see why, first recall that $f$-routing which is $\epsilon_1$-correct on $1$ instances has that there exists a family of channels $\{D^{x,y}_{M'\rightarrow Q}\}_{x,y}$ such that
    \begin{align*}
        \text{when}\,\,\, f(x,y)=1,\,\,\,F(\mathcal{D}^{x,y}_{M'\rightarrow Q} \circ\tr_{M} \circ\mathcal{N}^{x,y}_{Q\rightarrow MM'}(\Psi^+_{\bar{Q}Q}),\Psi^+_{\bar{Q}Q}) \geq 1-\epsilon_1\enspace.
    \end{align*}
    Define 
    \begin{align*}
        \sigma_{\bar{Q}Q} = \mathcal{D}^{x,y}_{M'\rightarrow Q}(\rho_{\bar{Q}M'})\enspace,
    \end{align*} 
    and use that by data processing, 
    \begin{align*}
        I(\bar{Q}:M')_{\rho_{\bar{Q}M}} \geq I(\bar{Q}:Q)_{\sigma_{\bar{Q}Q}}\enspace.
    \end{align*}
    But then by the Fuchs--van de Graaf inequalities $\sigma_{\bar{Q}Q}$ is close in trace distance to $\Psi^+_{\bar{Q}Q}$,
    \begin{align*}
        \frac{1}{2}||\Psi^+_{\bar{Q}Q} - \sigma_{\bar{Q}Q}||_1\leq \sqrt{2\epsilon_1 - \epsilon_1^2} \equiv \mu \enspace,
    \end{align*}
    and by continuity of the conditional entropy, 
    \begin{align*}
        I(\bar{Q}:Q)_{\Psi^+} - I(\bar{Q}:Q)_{\sigma_{\bar{Q}Q}} &= S(\bar{Q}|Q)_{\Psi^+} - S(\bar{Q}|Q)_{\sigma} \nonumber \\
        &\leq 2 \mu \log d_{\bar{Q}} + (1+\mu) h\left( \frac{\mu}{1+\mu}\right) \enspace,
    \end{align*}
    so that
    \begin{align*}
        2\log d_{\bar{Q}} - 2 \mu \log d_{\bar{Q}} + (1+\mu) h\left( \frac{\mu}{1+\mu}\right) \leq I(\bar{Q}:Q)_{\sigma} \enspace.
    \end{align*}
    We have that $\sigma$ is not a tensor product whenever the left hand side of the equation above is strictly positive. 
    For $d_Q=2$, this occurs for $\epsilon_1 > \epsilon_1^* \approx 0.08$, and $\epsilon_1^*$ approaches $1$ as $d_Q\rightarrow \infty$. 
\end{proof}

Next, we define what we call a \emph{structure function} for a protocol. 
The structure function is computed from the mid-protocol density matrix $\rho_{\bar{Q}MM'}$, and (as we will see) captures the zeros in $f(x,y)$.
\begin{definition}\label{def:structurefunction}
    Given an $f$-routing protocol with mid-protocol density matrix $\rho_{\bar{Q}MM'}$, define the \textbf{structure function} $g(x,y)$ according to
    \begin{align*}
        g(x,y) = \tr(\rho_{\bar{Q}M'}-\frac{\mathcal{I}}{d_{\bar{Q}}}\otimes \rho_{M'})^2\enspace.
    \end{align*}
\end{definition}
We claim that $g(x,y)$ captures some aspect of the structure in the function $f$ which must be present in a correct $f$-routing protocol. 
More concretely we have the following. 
\begin{lemma}
    In a perfectly correct $f$-routing protocol, the structure function $g(x,y)$ is zero if and only if $f(x,y)=0$.
\end{lemma}
\begin{proof}\,
    This follows because $g(x,y)=0$ is zero if and only if $\rho_{\bar{Q}M'}=\frac{\mathcal{I}}{d_{\bar{Q}}}\otimes \rho_{M'}$, and \cref{lemma:f=0tensorproduct} shows we have this tensor product form if and only if $f(x,y)=0$. 
\end{proof}

Our next job is to relate the function $g(x,y)$ to the entanglement available to Alice and Bob. 
We prove the following lemma. 
\begin{lemma}
    An $f$-routing protocol that uses a resource system with Schmidt rank $d_E$ has a structure function of the form
    \begin{align*}
        g(x,y) = \sum_{I=1}^{d_E^4} f_I(x)f_I'(y)\enspace.
    \end{align*}
\end{lemma}
\begin{proof}\,
    From the general form of a non-local quantum computation protocol, the density matrix $\rho_{\bar{Q}M_0'M_1'}(x,y)$ can be expressed as
    \begin{align*}
        \rho_{\bar{Q}M_0'M_1'}=\mathcal{N}^x_{QL\rightarrow M_0'}\otimes \mathcal{M}^y_{\bar{Q}\rightarrow M_1'}(\Psi^+_{\bar{Q}Q}\otimes \ketbra{\Psi}{\Psi}_{LR})\enspace.
    \end{align*}
    We will write $\ket{\Psi}_{LR}$ in the Schmidt basis, 
    \begin{align*}
        \ket{\Psi}_{LR} = \sum_{i=1}^{d_E} \ket{i}_L\ket{i}_R\enspace,
    \end{align*}
    with un-normalized vectors $\ket{i}_L, \ket{i}_R$. 
    Then we get
    \begin{align*}
        \rho_{\bar{Q}M_0'M_1'}&=\sum_{i,j=1}^{d_E}\mathcal{N}^x_{QL\rightarrow M_0'}(\Psi^+_{Q\bar{Q}}\otimes \ketbra{i}{j}_L)\otimes \mathcal{M}^y_{R\rightarrow M_1'}(\ketbra{i}{j}_{R}) \nonumber \\
        &= \sum_{i,j=1}^{d_E} A^{x,i,j}_{\bar{Q}M_0'}\otimes B^{y,i,j}_{M_1'}\enspace.
    \end{align*}
    We can also compute the trace over $\bar{Q}$ where we define $A$ and $B$ in the second line of the previous equation and get
    \begin{align*}
        \rho_{M_0'M_1'} &= \sum_{i,j=1}^{d_E} A^{x,i,j}_{M_0'}\otimes B^{y,i,j}_{M_1'}\enspace.
    \end{align*}
    Next, we compute $g(x,y)$. 
    It is convenient to make use of \cref{clm:trace-sq} to write
    \begin{align*}
        g(x,y) = \tr(\rho_{\bar{Q}M'}^2 - \frac{\mathcal{I}}{d_{\bar{Q}}^2}\otimes  \rho_{M'}^2)\enspace.
    \end{align*}
    Inserting the forms of $\rho_{\bar{Q}M'}$ and $\rho_{M'}$ into this, we obtain
    \begin{align*}
        g(x,y) &= \sum_{i,j,i',j'=1}^{d_E} \tr\left(\left( A^{x,i,j}_{QM_0'}A^{x,i',j'}_{QM_0'}  - \frac{\mathcal{I}}{d_Q^2}\otimes A^{x,i,j}_{M_0'}A^{x,i',j'}_{M_0'}\right)\otimes B^{y,i,j}_{M_1'}B^{y,i',j'}_{M_1'}\right) \nonumber \\
        &= \sum_{I=1}^{d_E^4} f_I(x) f'_I(y)\enspace,
    \end{align*}
    as needed. 
\end{proof}

We can view $g(x,y)$ as a matrix, and $f_I(x)$, $f_I'(y)$ as vectors so that the minimal number of terms appearing in this sum is the rank of the matrix $g(x,y)$. 
Thus we obtain the lower bound
\begin{align}\label{eq:rankbound}
    \boxed{\FR_0(f) \geq \frac{1}{4} \log \rank g|_f\enspace.}
\end{align}
We know that $g$ is zero and non-zero for the same inputs as $f$, for which we have given the $g|_f$ notation as a reminder. In some cases, this is enough to let us lower bound the rank of $g$, as we discuss below. 

We can also notice that we can reverse the role of $M$ and $M'$, and assume perfect correctness on $1$ instances, leading to a similar bound, 
\begin{align}\label{eq:rankboundnegation}
    \boxed{\FR_1(f) \geq \frac{1}{4} \log \rank g|_{\neg f}\enspace,}
\end{align}
where $\neg f$ is the negation of $f$. 

We can also rephrase the above bounds in terms of the non-deterministic quantum communication complexity. 
Observe that 
\begin{align}\label{eq:gfandnrank}
    \rank g|_f \geq \text{nrank}(f)\enspace. 
\end{align}
Because the non-deterministic rank considers matrices with arbitrary complex values where $f(x,y)=1$, while $g|_f$ has the further constraint that it has positive real values when $f(x,y)=1$, it is possible there are functions for which $\rank g|_f > \text{nrank}(f)$.
We have not looked for such examples. 

Using \cref{eq:gfandnrank} and \cref{thm:QNPandnrank}, we obtain
\begin{align*}
    \FR_0(f) &\geq \frac{1}{4}(\QNP^{\cc}(f) -1)\enspace, \nonumber \\
    \FR_1(f) &\geq \frac{1}{4}(\coQNP^{\cc}(f) - 1)\enspace.
\end{align*}
Further, for $f$-routing which has two sided perfection, we have
\begin{align*}
    \text{p}\FR(f) \geq \frac{1}{4}\max\{\QNP^{\cc}(f), \coQNP^{\cc}(f) \} -1\enspace.
\end{align*}
Or, labelling the set of function families which can be $f$-routed on with polylogarithmic entanglement by $pFR$, we have
\begin{align*}
    \text{p}\FR \subseteq \QNP^{\cc} \cap \coQNP^{\cc}\enspace.
\end{align*}
In the classical setting, $\NP^{\cc}\cap \coNP^{\cc}=\text{P}^{\cc}$ \cite{babai1986complexity}, but (to our knowledge) it is not known if the quantum analogue holds. 

%%%%%%%%%%%%%%%%%%%%%%%%%%%%%%%%%%%%%%%%%%%%
\subsection{\texorpdfstring{$f$}{TEXT}-BB84 bounds}
%%%%%%%%%%%%%%%%%%%%%%%%%%%%%%%%%%%%%%%%%%%%

We first define the $f$-BB84 task formally.
\begin{definition}\label{def:qubitfbb84}
    A \textbf{qubit $f$-BB84} task is defined by a choice of Boolean function $f:\{ 0,1\}^{2n}\rightarrow \{0,1\}$, and a $2$ dimensional Hilbert space $\mathcal{H}_Q$.
    Inputs $x\in \{0,1\}^{n}$ and system $Q$ are given to Alice, and input $y\in \{0,1\}^{n}$ is given to Bob.
    The $Q$ system is in the maximally entangled state with a reference $\bar{Q}$.
    Alice and Bob exchange one round of communication, with the combined systems received or kept by Bob labelled $M$ and the systems received or kept by Alice labelled $M'$.
    Define projectors
    \begin{align*}
        \Pi^{q,b}=H^q\ketbra{b}{b}H^q \enspace.
    \end{align*}
    The qubit $f$-BB84 task is completed $\epsilon$-correctly on input $(x,y)$ if there exist POVM's $\{\Lambda^{x,y,0}_{M},\Lambda^{x,y,1}_{M}\}$, $\{\Lambda^{x,y,0}_{M'},\Lambda^{x,y,1}_{M'}\}$ such that,
    \begin{align*}
        \tr(\Pi^{f(x,y),b}_{\bar{Q}}\otimes \Lambda^{x,y,b}_{M} \otimes \Lambda^{x,y,b}_{M} \rho_{\bar{Q}MM'}) \geq 1-\epsilon \enspace.
    \end{align*}
    That is, Alice and Bob succeed when they both obtain the same outcome as the referee.
    We will say a protocol is $(\epsilon_0,\epsilon_1)$-correct if the above holds with $\epsilon\leq \epsilon_0$ for all $(x,y)\in f^{-1}(0)$ and with $\epsilon\leq \epsilon_1$ for all $(x,y)\in f^{-1}(1)$. 
\end{definition}

As in the analysis of $f$-routing, we look for a function of the density matrix prepared by the protocol which captures the structure of the function $f(x,y)$. 
We define this next. 
\begin{definition}
    Given a $f$-BB84 protocol with mid-protocol density matrix $\rho_{\bar{Q}MM'}$, define the structure function by first defining the density matrices
    \begin{align*}
        \rho^0_M &= \bra{0}_{\bar{Q}} \rho_{\bar{Q}M} \ket{0}_{\bar{Q}}\enspace, \nonumber \\
        \rho^1_M &= \bra{1}_{\bar{Q}} \rho_{\bar{Q}M} \ket{1}_{\bar{Q}}\enspace.
    \end{align*}
    These are the post measurement states on $M$ when the referee measures $R$ in the computational basis, and finds outcome $b=0,1$.
    We similarly define $\rho^0_{M'}$ and $\rho^1_{M'}$.
    Then define
    \begin{align*}
        g(x,y) = \tr(\rho_{M}^0\rho_M^1) + \tr(\rho_{M'}^0\rho_{M'}^1)\enspace.
    \end{align*}
\end{definition}
Our goal is to show that $g(x,y)$ and $f(x,y)$ are zero on exactly the same inputs. 
To do this, we need some results from earlier analyses of $f$-BB84. 
In particular we recall the following definition from \cite{bluhm2021position,asadi2024linear} (our notation follows \cite{asadi2024linear}). 
\begin{definition}\label{def:01setBB84}
    For $\epsilon \in [0,1]$, let $S^{\epsilon}_{0}$ be the set of states $\rho_{\bar{Q}MM'}$ such that there exists a measurement on subsystem $M$ and a measurement on subsystem $M'$ that each allow us to guess the outcome of a measurement in the computational basis on $\bar{Q}$ with probability at least $1-\epsilon$. 
    Similarly, let $S^{\epsilon}_{1}$ be the set of states $\rho_{\bar{Q}MM'}$ such that there exists a measurement on subsystem $M$ and a measurement on subsystem $M'$ that allows us to guess the outcome of a measurement in the Hadamard basis on $\bar{Q}$ with probability at least $1-\epsilon$. 
\end{definition}
With this definition we can state the following result from \cite{bluhm2021position,asadi2024linear}. 
\begin{lemma}\label{lemma:fBB84emptyintersection}
    Suppose that $\epsilon<0.11$. 
    Then $S^{\epsilon}_{0} \cap S^{\epsilon}_{1} = \emptyset$.
\end{lemma}
This result says that a state which works well in $1$ instances cannot also be a state that works well in $0$ instances, and vice versa. 
We use this idea to prove the next lemma. 
\begin{lemma}
    In a $f$-BB84 protocol with $\epsilon_0=0$ and $0\leq \epsilon_1<0.11$, the structure function $g(x,y)$ is zero if and only if $f(x,y)$ is zero. 
\end{lemma}
\begin{proof}\,
    We first establish the following fact: $g(x,y)=0$ if and only if both the following are true: $\rho_M^0$ and $\rho_M^1$ have non-overlapping support, and $\rho_{M'}^0$ and $\rho_{M'}^1$ have non-overlapping support. 
    To see this, first assume $g(x,y)=0$.
    Then diagonalize $\rho^0_M$ and $\rho_M^1$ so that
    \begin{align*}
        \rho^0_M &= \sum_{i} \lambda_i\ketbra{\psi_i}{\psi_i}\enspace, \nonumber \\
        \rho^1_M &= \sum_{j} \mu_j\ketbra{\phi_j}{\phi_j}\enspace,\nonumber \\
        \rho^0_{M'} &= \sum_{i} \lambda'_i\ketbra{\psi'_i}{\psi'_i}\enspace, \nonumber \\
        \rho^1_{M'} &= \sum_{j} \mu'_j\ketbra{\phi'_j}{\phi'_j}\enspace,
    \end{align*}
    with $\lambda_i, \mu_j, \lambda_j', \mu_j'>0$. 
    Then $0=g(x,y)$ gives that
    \begin{align*}
        0 =\tr(\rho^0_M\rho^1_M) + \tr(\rho^0_{M'}\rho^1_{M'})= \sum_{i,j} \lambda_i \mu_j |\braket{\phi_j}{\psi_i}|^2 + \sum_{i,j} \lambda_i' \mu_j' |\braket{\phi'_j}{\psi'_i}|^2\enspace.
    \end{align*}
    But this occurs if and only if 
    \begin{align*}
        \forall i,j,\,\,\,\,\, \braket{\phi_j}{\psi_i}=0\enspace,
    \end{align*}
    and
    \begin{align*}
        \forall i,j,\,\,\,\,\, \braket{\phi'_j}{\psi'_i}=0\enspace,
    \end{align*}
    so that $\rho_M^0$ and $\rho^1_M$ have non-overlapping support, and $\rho_{M'}^0$ and $\rho_{M'}^1$ do, as needed. 
    For the converse, a direct calculation shows both $\tr(\rho^0_M\rho^1_M)$ and $\tr(\rho^0_{M'}\rho^1_{M'})$ are 0. 

    Now, suppose that $f(x,y)=0$. 
    By assumption the protocol has $\epsilon_0=0$.
    Considering Alice, this gives that Alice can obtain $b$ from system $M$ with probability $1$. 
    But this says that $\rho_M^0$ and $\rho_M^1$ must have orthogonal support. 
    Considering Bob, perfect correctness also gives that Bob can determine $b$ from $M'$ with probability $1$, so that $\rho_{M'}^0$ and $\rho_{M'}^1$ have orthogonal support. 
    Then, from the comment above, we have that $g(x,y)=0$.

    Next suppose that $f(x,y) = 1$. 
    By assumption, this gives that $\rho_{\bar{Q}MM'}\in S_1^\epsilon$, so by \cref{lemma:fBB84emptyintersection} that $\rho_{\bar{Q}MM'}\notin S^\epsilon_0$. 
    This means that either $\rho_M^0$ and $\rho_M^1$ have overlapping support, $\rho_{M'}^0$ and $\rho_{M'}^1$ have overlapping support, or they both do. 
    This is because otherwise if they both have non-overlapping support there would be a measurement to perfectly determine $b$ from both $M$ and $M'$, and the state would be in $S_0^\epsilon$. 
    But then by our initial comment that $g(x,y)=0$ if and only if these density matrices have orthogonal support, we must have $g(x,y)> 0$. 
\end{proof}

\begin{lemma}
    An $f$-BB84 protocol that uses a resource system with Schmidt rank $d_E$ has a structure function of the form
    \begin{align*}
        g(x,y) = \sum_{I=1}^{2d_E^4} f_I(x)f_I'(y)\enspace.
    \end{align*}
\end{lemma} 
\begin{proof}\,
From the general form of a non-local quantum computation protocol, the density matrix $\rho_{\bar{Q}M_0'M_1'}(x,y)$ can be expressed as
\begin{align*}
    \rho_{\bar{Q}M_0'M_1'}=\mathcal{N}^x_{QL\rightarrow M_0'}\otimes \mathcal{M}^y_{R\rightarrow M_1'}(\Psi_{\bar{Q}Q}\otimes \ketbra{\Psi}{\Psi}_{LR})\enspace.
\end{align*}
We will write $\ket{\Psi}_{LR}$ in the Schmidt basis, 
\begin{align*}
    \ket{\Psi}_{LR} = \sum_{i=1}^{d_E} \ket{i}_L\ket{i}_R\enspace,
\end{align*}
with un-normalized vectors $\ket{i}_L, \ket{i}_R$. 
Then we get
\begin{align*}
    \rho_{\bar{Q}M_0'M_1'}&=\sum_{i,j=1}^{d_E}\mathcal{N}^x_{QL\rightarrow M_0'}(\Psi^+_{\bar{Q}Q}\otimes \ketbra{i}{j}_L)\otimes \mathcal{M}^y_{R\rightarrow M_1'}(\ketbra{i}{j}_{R}) \nonumber \\
        &= \sum_{i,j=1}^{d_E} A^{x,i,j}_{\bar{Q}M_0'}\otimes B^{y,i,j}_{M_1'}\enspace.
\end{align*}
We can also project $\bar{Q}$ into $\ket{b}_{\bar{Q}}$, $b\in \{0,1\}$ to get $\rho_{M'}^i$, 
\begin{align*}
    \rho_{M_0'M_1'}^b &= \sum_{i,j=1}^{d_E} A^{b,x,i,j}_{M_0'}\otimes B^{y,i,j}_{M_1'}\enspace.
\end{align*}
A similar calculation gives $\rho_M^{b}$. 
Inserting these expressions into $g(x,y)$, we obtain an expression of the form
\begin{align*}
    g(x,y) = \sum_{I=1}^{2d_E^4} f_I(x) f'_I(y)\enspace.
\end{align*}
where the $d_E^4$ appears because each of the two trace terms in $g(x,y)$ comes with a product of two density matrices, each of which involves two sums running from $1$ to $d_E$. 
We collect all of these sums into the index $i$ appearing here, which now runs 1 to $2 d_E^4$. 
\end{proof}

Similarly to the case of $f$-routing, the above lemma implies that
\begin{align}\label{eq:fBB84bound}
    \boxed{\FBB84_{0}(f) \geq \frac{1}{4}( \log \rank g|_f-1)\enspace.}
\end{align}
where the notation $g|_f$ indicates as before that $g$ is zero and non-zero on the same inputs as $f(x,y)$. 
We can also obtain a lower bound on $\FBB84_{1}(f)$ by considering the negated function, 
\begin{align}
    \boxed{\FBB84_{1}(f) \geq \frac{1}{4}( \log \rank g|_{\neg f}-1)\enspace.}
\end{align}
Also similarly to the case of $f$-routing, we can use \cref{eq:gfandnrank} and \cref{thm:QNPandnrank}, we obtain
\begin{align*}
     \FBB84_0(f) &\geq \frac{1}{4}(\QNP^{\cc}(f) -2)\nonumber \\
     \FBB84_1(f) &\geq \frac{1}{4}(\coQNP^{\cc}(f) -2)
\end{align*}

%%%%%%%%%%%%%%%%%%%%%%%%%%%%%%%%%%%%%%%%%%%%
\section{A new lower bound technique for CDS}
%%%%%%%%%%%%%%%%%%%%%%%%%%%%%%%%%%%%%%%%%%%%

It was recently understood \cite{allerstorfer2023relating} that $f$-routing is closely related to a classical information theoretic primitive known as conditional disclosure of secrets. 
As a consequence, our new lower bounds on $f$-routing imply lower bounds on CDS.
In this work, we will informally describe these lower bounds, but see \cite{asadi2024conditional} for a more complete treatment. 

The \emph{conditional disclosure of secrets} primitive involves three parties, Alice, Bob, and the referee. 
Alice receives input $x\in \{0,1\}^n$, and Bob receives $y\in \{0,1\}^n$.
Alice additionally holds a secret string $z$, and Alice and Bob share a random string $r$. 
Alice and Bob do not communicate, but each sends a message to the referee. 
The goal is for the referee to be able to learn $z$ if and only if a function $f(x,y)=1$, where the choice of the function $f$ is known to all parties, and the referee knows both inputs $x$ and $y$.

The quantum generalization of CDS, \emph{conditional disclosure of quantum secrets} or CDQS, was recently defined \cite{allerstorfer2023relating}. 
In the quantum setting, Alice and Bob share a quantum state (rather than a random string) and can send quantum messages to the referee. 
The secret can now also be taken to be a quantum system $Q$ rather than a classical string. 
The security and privacy conditions are similar to the classical case: correctness means that the referee can invert the map on $Q$ applied by Alice and Bob when $f(x,y)=1$, and security means the map applied by Alice and Bob is close to completely depolarizing when $f(x,y)=0$. 
In general, we allow small correctness and privacy errors. 
See \cite{allerstorfer2023relating} for a formal definition. 

We define the minimal number of qubits in Alice and Bob's shared resource system in a CDQS scheme for function $f$ by $\CDQS(f)$. 
If we require the protocol to be perfectly correct we write $\pc{\CDQS}(f)$, and when perfectly private we write $\pp{\CDQS}(f)$. 
Similarly, for classical CDS, we use ${\CDS}(f)$ to denote the randomness cost of classical CDS.
Again we modify this to $\pc\CDS(f)$, $\pp\CDS(f)$ for perfectly correct and perfectly private CDS randomness cost.\footnote{The reader should be warned that the same notation is used to denote the \emph{communication} cost in other papers on CDS.} 

In \cite{allerstorfer2023relating}, it was shown that a good classical CDS protocol for a function $f$ implies a good quantum CDS protocol for the same function. 
This leads to\footnote{In fact, to obtain this we also need an amplification result which appears in \cite{asadi2024conditional}. These bounds are proved formally there.}
\begin{align*}
    \CDS(f) = \Omega(\CDQS(f))\enspace,
\end{align*}
and the same relationships hold for the perfectly correct and perfectly private variants.
Also in \cite{allerstorfer2023relating}, it was shown that a good CDQS leads to a good $f$-routing protocol, and vice versa. 
This leads to
\begin{align*}
    \CDQS(f) = \Theta(\FR(f))\enspace, 
\end{align*}
Also, the perfectly private and perfectly correct variants are related to $f$-routing with one-sided errors, 
\begin{align*}
    \pp\CDQS(f) = \Theta(\FR_0(f))\enspace, \nonumber \\
    \pc\CDQS(f) = \Theta(\FR_1(f))\enspace.
\end{align*}
This implies our lower bounds from \cref{eq:rankbound} and \cref{eq:rankboundnegation} apply to CDS, 
\begin{align*}
    \pp\CDQS(f) = \Omega(\log \rank{ g|_f})\enspace, \nonumber \\
    \pc\CDQS(f) = \Omega(\log \rank{ g|_{\neg f}})\enspace.
\end{align*}
This also gives lower bounds in terms of $\QNP^{\cc}(f)$ and $\coQNP^{\cc}(f)$, and the upper bound can be written in terms of classical CDS. 
In this way, we can obtain the lower bound
\begin{align*}
    \boxed{\pc\CDS(f) \geq \Omega(\coQNP^{\cc}(f))\enspace.}
\end{align*}
This is similar to a lower bound obtained in \cite{applebaum2021placing}, 
\begin{align*}
    \pc\CDS(f) \geq \frac{1}{2} \coNP^{\cc}(f)\enspace.
\end{align*}
The lower bound on perfectly private CDS randomness cost seems to be new, 
\begin{align*}
     \boxed{\pp\CDS(f) \geq \Omega(\QNP^{\cc}(f))\enspace.}
\end{align*}
For perfect CDS, we can bound the randomness cost by the maximum of the $\QNP^{\cc}$ and $\coNP^{\cc}$ cost. 
If we let $\text{p}\CDS$ denote the set of families of functions for which CDS can be implemented with perfect correctness and perfect privacy using polylogarithmic randomness.
Then we obtain
\begin{align*}
    \text{p}\CDS \subseteq \QNP^{\cc} \cap \coNP^{\cc}\enspace.
\end{align*}

%%%%%%%%%%%%%%%%%%%%%%%%%%%%%%%%%%%%%%%%%%%%
\section{Application to concrete functions}
%%%%%%%%%%%%%%%%%%%%%%%%%%%%%%%%%%%%%%%%%%%%

We detail lower bounds for equality, non-equality, greater than and less than, set-disjointness, and set-intersection below. 
We were not able to use our technique to lower bound inner product.

Our results are summarized succinctly in \cref{table:FRbounds}. 

\begin{figure}
\begin{center}
\begin{tabular}{ |c c c c c|}
\hline
Function & pFR & $\FR_1$ & $\FR_0$ & FR \\ 
 \hline
 \hline
 Equality & ${\Theta(n)}$ & $\Theta(1)$ & ${\Theta(n)}$ & $\Theta(1)$ \\  
 Non-Equality & $\Theta(n)$ & $\Theta(n)$ & $\Theta(1)$ & $\Theta(1)$ \\  
 Inner-Product &  $O(n)$, $\Omega(1)$ & $O(n)$, $\Omega(1)$ & $O(n)$, $\Omega(1)$ & $O(n)$, $\Omega(1)$ \\  
 Greater-Than & $\Theta(n)$ & $\Theta(n)$ & $\Theta(n)$ & $O(n)$, $\Omega(1)$ \\  
 Set-Intersection & $\Theta(n)$ & $\Theta(n)$ & $O(n)$,$\Omega(1)$ & $O(n)$, $\Omega(1)$ \\  
 Set-Disjointness & $\Theta(n)$ & $O(n)$,$\Omega(1)$ & $\Theta(n)$ & $O(n)$, $\Omega(1)$ \\  
 \hline
\end{tabular}
\end{center}
\caption{Known upper and lower bounds on entanglement cost (quantified by the log Schmidt rank) in $f$-routing for natural functions. No non-trivial lower bounds on entanglement were known prior to this work. Linear upper bounds on all of these functions follow from \cite{buhrman2013garden, cree2023code}. The $\Omega(1)$ lower bounds follow from \cref{appendix:trivialbounds}.}\label{table:FRbounds}
\end{figure}

%%%%%%%%%%%%%%%%%%%%%%%%%%%%%%%%%%%%%%%%%%%%
\vspace{0.2cm}
\noindent \textbf{Equality and non-equality}
\vspace{0.2cm}
%%%%%%%%%%%%%%%%%%%%%%%%%%%%%%%%%%%%%%%%%%%%

Choose $f(x,y)$ to be the equality function,
\begin{align*}
EQ(x,y) = \begin{cases}
		0, & x\neq y\\
            1, & x=y
		 \end{cases}\enspace.
\end{align*}
Then $g(x,y)$ is zero except on the diagonal, which forces it to have full rank, so from \cref{eq:rankbound} and \cref{eq:fBB84bound}
\begin{align*}
    \FR_0(EQ) \geq \frac{n}{4}, \,\,\,\,\,\,\,\,\,\,\,\, \text{p}\FBB84(EQ) \geq \frac{1}{4}(n-1)\enspace.
\end{align*}
For the not-equals function, we can use bound from \cref{eq:rankboundnegation} to obtain
\begin{align*}
    \FR_1(NEQ) \geq \frac{n}{4}\enspace.
\end{align*}

%%%%%%%%%%%%%%%%%%%%%%%%%%%%%%%%%%%%%%%%%%%%
\vspace{0.2cm}
\noindent \textbf{Greater and less than}
\vspace{0.2cm}
%%%%%%%%%%%%%%%%%%%%%%%%%%%%%%%%%%%%%%%%%%%%

Similarly, the `greater than' function, 
\begin{align*}
    GT(x,y) = \begin{cases}
			0, & x< y\\
                1, & x\geq y
		 \end{cases}\enspace.
\end{align*}
is upper triangular with non-zero elements on the diagonal, so it also has full rank, and we obtain a linear lower bound. 
\begin{align*}
    \FR_0(GT) \geq \frac{n}{4}, \,\,\,\,\,\,\,\,\,\,\,\, \text{p}\FBB84(GT)\geq \frac{1}{4}(n-1)\enspace.
\end{align*}
Further, because the negation of Greater-Than is also full rank, we can also bound $fR_1$, 
\begin{align*}
    \FR_1(GT)\geq \frac{n}{4}\enspace.
\end{align*}
The same bounds hold for the `less than' function.

%%%%%%%%%%%%%%%%%%%%%%%%%%%%%%%%%%%%%%%%%%%%
\vspace{0.2cm}
\noindent \textbf{Set-Intersection and set-disjointness}
\vspace{0.2cm}
%%%%%%%%%%%%%%%%%%%%%%%%%%%%%%%%%%%%%%%%%%%%

Set disjointness is upper left triangular.\footnote{To see why, consider that on the diagonal of the truth table running top right to bottom left, we have that $x+y=11...11$, the all $1$'s string. This means $x$ and $y$ must have non-zero values in non-overlapping locations, e.g. for two bits this diagonal consists of $(x=00,y=11)$, $(x=01,y=10)$, $(x=10,y=01)$ and $(x=11,y=00)$. Moving downward from any entry on that diagonal $y$ becomes larger, so must now have an overlapping entry with the $x$ string.}
This implies it is full rank, therefore
\begin{align*}
    \FR_0(DISJ) \geq \frac{n}{4},\,\,\,\,\,\,\,\,\,\,\,\text{p}\FBB84(DISJ) \geq \frac{1}{4}(n-1)\enspace.
\end{align*}
Since the negation of set intersection is set disjointness and hence of full rank, we also obtain
\begin{align*}
    \FR_1(INT) \geq \frac{n}{4}\enspace.
\end{align*}

%%%%%%%%%%%%%%%%%%%%%%%%%%%%%%%%%%%%%%%%%%%%
\section{Discussion}
%%%%%%%%%%%%%%%%%%%%%%%%%%%%%%%%%%%%%%%%%%%%

In this article, we have given the first non-trivial lower bounds on entanglement in the $f$-routing and $f$-BB84 tasks. 
For some natural functions, including equality, greater than, and set-disjointness, our lower bounds are linear in the classical input size. 
This realizes a long standing aspiration for $f$-routing and $f$-BB84 based position-verification schemes: that an honest party can compute a simple, classical function and perform $O(1)$ quantum operations, while a dishonest party must use entanglement that grows with the size of the classical input strings. 
However, this result comes at the unfortunate cost of requiring (one-sided) perfect correctness for the lower bounds to apply. 
Nonetheless, even proving bounds in this perfect case has so far been elusive, and we hope that our progress may open future avenues toward robust linear lower bounds. 

An interesting observation is that our bounds lead to relations among operational quantities, e.g.
\begin{align*}
    \FR_0(f) \geq \QNP^{cc}(f)\enspace
\end{align*}
but are derived by separately relating $\FR_0$ and $\QNP^{cc}$ to the non-deterministic rank. 
Given this result, it would be natural to expect a reduction at an operational level from $f$-routing to quantum non-deterministic communication complexity, but we know of no such reduction. 
Another comment related to the non-deterministic rank is that our lower bounds are in terms of the rank of matrices with positive entries wherever $f(X,y)=1$, and $0$ entires otherwise. 
Minimizing over all such matrices produces a number lower bounded by the non-deterministic rank (which allows arbitrary complex non-zero entries). 
It would be interesting to understand if the minimization over positive entries is ever larger than over complex entries.

As a consequence of our lower bounds, we also obtain a new lower bound technique for randomness complexity in CDS with either perfect privacy or perfect correctness.
We highlight the bound 
\begin{align*}
     \boxed{\pp\CDS(f) \geq \Omega(\QNP^{\cc}(f))\enspace.}
\end{align*}
which appears to be new and incomparable to earlier bounds on $\pp\CDS$. 
It would be interesting to construct explicit functions for which this lower bound is tighter than earlier ones, for instance, bounds from the $\text{PP}^{\cc}$ cost of $f$ \cite{applebaum2020power}. 

\vspace{0.2cm}
\noindent \textbf{Note added:} Recently, an equivalence between $f$-routing and $f$-BB84 was shown in \cite{bluhm2025complexity}. This allows our rank lower bounds on $f$-BB84 to be derived from the rank lower bounds on $f$-routing.

\vspace{0.2cm}
\noindent \textbf{Acknowledgements:} We thank Richard Cleve for helpful discussions during the development of this project. 
Harry Buhrman and François Le Gall suggested a connection between our bounds and non-deterministic rank, and consequently to non-deterministic quantum communication complexity.
Research at Perimeter Institute is supported in part by the Government of Canada through the Department of Innovation, Science and Economic Development Canada and by the Province of Ontario through the Ministry of Colleges and Universities.

\appendix

%%%%%%%%%%%%%%%%%%%%%%%%%%%%%%%%%%%%%%%%%%%%
\section{\texorpdfstring{$\Omega(1)$}{TEXT} lower bounds on entanglement in \texorpdfstring{$f$}{TEXT}-routing}\label{appendix:trivialbounds}
%%%%%%%%%%%%%%%%%%%%%%%%%%%%%%%%%%%%%%%%%%%%

It is straightforward to establish an $\Omega(1)$ lower bound on entanglement for $f$-routing, though to our knowledge this lower bound is not recorded elsewhere in the literature, so we give it here. 
Note that for $f$-BB84 a $\Omega(1)$ lower bound follows from theorem 6.1 in \cite{buhrman2014position}. 

To show the lower bound in the $f$-routing setting, we will first notice that in a $f$-routing task we can always replace a mixed resource state $\Psi_{LR}$ with a pure state $\ket{\Psi}_{LR}$. 
Our lower bound will apply whenever $f(x,y)$ has an input $x=x_0$ such that $f(x_0,y)=F(y)$ is non-constant. 
We will assume perfect correctness of the $f$-routing protocol in our arguments but the generalization to the robust setting is immediate. 

Considering the perfect protocol, we allow Alice and Bob to complete the easier task of performing $f$-routing when they know $x=x_0$. 
Alice, on the left, applies a quantum channel $\mathcal{N}_{QL\rightarrow AB}$, she keeps system $A$ and sends $B$ to Bob. 
Without decreasing their success probability we can take the isometric extension of this channel and have Alice keep the purifying system so that the state on $ABR$ is pure. 
Bob, who knows $y$, knows where system $Q$ should be brought so can, without decreasing their success probability, always simply forward the $R$ system to whichever party should receive $Q$. 
Correctness of the protocol then requires that both $AR$ and $BR$ can recover $Q$. 

Let $Q$ be in the maximally entangled state $\Psi^+_{\bar{Q}Q}$ with reference system $\bar{Q}$. 
Then correctness of the protocol implies that
\begin{align}\label{eq:recoveryMIs}
    I(\bar{Q}:AR) = 2\log d_R\enspace, \nonumber \\
    I(\bar{Q}:BR) = 2\log d_R\enspace.
\end{align}
We claim this implies that $S(R)\geq \log d_R$. 
First notice that purity of $\bar{Q}ABR$ and the second line above implies $I(\bar{Q}:A)=0$.  
Then we apply the inequality, 
\begin{align*}
    I(\bar{Q}:AR) \leq I(\bar{Q}:A) + 2S(R)\enspace,
\end{align*}
which can be proven from strong subadditivity, and using the first line of \cref{eq:recoveryMIs} we are done. 
In the robust setting, we can prove a similar bound by applying the above argument along with the continuity of the von Neumann entropy.

\printbibliography

\end{document}